\PassOptionsToPackage{table}{xcolor}
\documentclass[letterpaper,twocolumn,10pt]{article}
\usepackage{usenix}

\usepackage{tikz}
\usepackage{amsmath}

\usepackage[most]{tcolorbox}

\definecolor{takeawaygreen}{RGB}{46,125,80}
\definecolor{takeawayback}{RGB}{244,250,246}

\newtcolorbox{takeawaybox}{
  enhanced,
  colback=takeawayback,
  colframe=takeawaygreen,
  boxrule=0.5pt,
  arc=2pt,
  left=6pt,
  right=6pt,
  top=5pt,
  bottom=5pt,
  before skip=6pt,
  after skip=6pt,
}
\usepackage{booktabs}  %
\usepackage{tabularx}  %
\usepackage{multirow}  %
\usepackage{amssymb}   %
\usepackage{dsfont}
\usepackage{amsfonts}
\usepackage{amsthm}
\usepackage{enumitem}
\usepackage[caption=false,font=footnotesize]{subfig}
\usepackage{array}
\usepackage{cancel}

\usepackage{pifont}
\newcommand{\cmark}{\ding{51}}
\newcommand{\xmark}{\ding{55}}

\providecommand{\resultbox}[2]{%
\par\smallskip
\noindent\begingroup
\renewcommand{\arraystretch}{1.15}%
\begin{tabular}{@{\hspace{0.2em}}|@{\hspace{0.75em}}p{0.90\linewidth}@{}}
\small\textbf{Takeaway.} \emph{#2}
\end{tabular}%
\endgroup
\par\smallskip
}

\newcommand{\best}[1]{\textbf{#1}}

\begin{document}

\date{}

\title{\Large \bf MESA: Prioritizing Security-Critical \\ Communication Channels in Multi-Agent Systems}

\author{
{\rm Kunyang Li, Kyle Domico, Jonathan Gregory, and Patrick McDaniel}\\[0.5ex]
{\rm University of Wisconsin--Madison}\\[0.5ex]
{\rm\texttt{\{kli253, jgregory8\}@wisc.edu}} \qquad
{\rm\texttt{\{domico, mcdaniel\}@cs.wisc.edu}}
}

\maketitle

\begin{abstract}
Today, multi-agent systems (MAS) are becoming increasingly capable by exchanging, critiquing, and aggregating intermediate information. However, this internal communication is both a source of capability and an emerging attack surface. Harmful content can be transformed and propagated through messages among legitimate agents. Applying costly defenses everywhere incurs latency and performance degradation. A question therefore arises: \textit{which inter-agent information flows should defenders prioritize?} In this paper, we introduce \textsc{Mesa}, a novel allocation layer that ranks internal communication channels by security criticality. \textsc{Mesa} requires no attack traces; it combines lightweight graph-theoretic features from the communication topology with workflow-aware features capturing runtime behavior from execution probes. Our study spans eight topologies, three application domains, five standard open-weight models, and diverse attack configurations. In customer service, monitoring the top 10\% of \textsc{Mesa}-ranked channels covers more than twice as many successful attacks as random allocation. Further, integrated with a runtime detector, \textsc{Mesa} prevents approximately 18.2 percentage points more attacks while maintaining 93.3\% accuracy. Results hold under adaptive attackers and transfer to LangGraph, a real-world orchestration framework. Together, these results show that security-critical information flows are often concentrated in MAS and can be identified without attack traces, enabling effective edge-level defense allocation. 

\end{abstract}

\section{Introduction}
\label{sec:intro}

LLM-based multi-agent systems (MAS) are increasingly deployed to automate complex workflows in high-stakes domains, including financial trading~\cite{xiaoTradingAgentsMultiAgentsLLM2024, liQuantAgentsMultiagentFinancial2025}, medical consultation~\cite{roseMEDDxAgentUnifiedModular2025,schmidgallAgentClinicMultimodalAgent2025}, and software development~\cite{antoniadesSWESearchEnhancingSoftware2025, hongMetaGPTMetaProgramming2024, qianChatDevCommunicativeAgents2024}. Their capabilities arise in large part from inter-agent communication~\cite{duImprovingFactualityReasoning2023,chunWhatAgentsCommunicate2026}. Specialized agents exchange evidence, critique intermediate outputs, and aggregate reasoning that no single agent holds in full. These interactions allow an MAS to divide work across agents with different roles, context, and privileges, but they also make system behavior depend on how information flows through the communication graph~\cite{wuAutoGenEnablingNextGen2023,LangGraphOverview,CrewAIDocumentationCrewAI}.

\begin{figure}[t]
    \centering
    \includegraphics[width=\columnwidth]{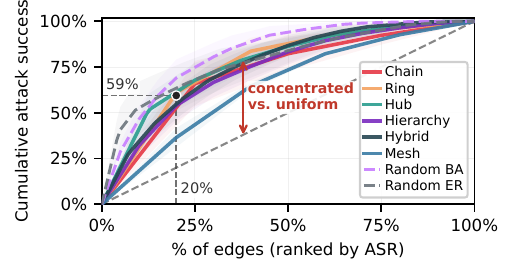}
    \caption{Attack impact is concentrated across communication channels, which are ordered by empirical per-edge attack success rate (ASR). The curve reports the cumulative share of successful attacks.}
    \label{fig:rq1-within-topology}
\end{figure}

This same inter-agent communication creates a system-level attack surface. Adversarial content can enter through user input, untrusted tool outputs, compromised agents, or corrupted communication infrastructure~\cite{debenedettiAgentDojoDynamicEnvironment2024,chenAgentPoisonRedteamingLLM2024,greshakeNotWhatYouve2023,triedmanMultiAgentSystemsExecute2025,heRedTeamingLLMMultiAgent2025, liA2ASecBenchProtocolAwareSecurity2025}.
Once introduced, such content does not stay local; it is transformed and propagated by otherwise legitimate components. 
Recent work~\cite{papadopoulosMindVirusesSelfPropagating2026,zhangAgentWormSelfPropagatingAttacks2026,xieSparkFireModeling2026} shows that a single injected message can cause autonomous, self-propagating compromise and system-wide false consensus in production frameworks. Authentication~\cite{A2AProtocol} can establish identity and transit integrity but not semantic trustworthiness~\cite{liA2ASecBenchProtocolAwareSecurity2025}. Content-level defenses~\cite{chennabasappaLlamaFirewallOpenSource2025,miaoBlindGuardSafeguardingLLMbased2025, shiSAIGuardCommunicationStateSimulation2026} (e.g., monitors and LLM judges) can be helpful; however, applying them indiscriminately to every interaction incurs additional inference, latency, and false positives~\cite{zhangCatchingOneFive2026, wangGSafeguardTopologyGuidedSecurity2025,wuMonitoringLLMbasedMultiAgent2025}. As MAS grow in size and complexity, the resulting question is not solely \textit{how} to defend their communication, but \textit{where} limited defense resources should be deployed.

In this paper, we hypothesize that some communication channels in MAS are more \textit{consequential} than others---that is, corrupting them is more likely to change the system outcome. Some channels carry unique, decision-critical information, while others are redundant or weakly coupled to the final decision. Empirically, we observe substantial imbalance: the top $20\%$ of channels account for more than half of successful attacks on average (\autoref{fig:rq1-within-topology}, \S\ref{sec:rq-preliminary}). This concentration creates opportunities for targeted defense, requiring effective identification of security-critical channels \textit{before} attacks happen. 

We introduce \textsc{Mesa} (MAS Edge Saliency Analysis), a novel allocation layer that ranks inter-agent communication channels by their expected security impact without observing attack outcomes. Predicting this imbalance requires capturing both \textit{where} a channel sits in the communication network and \textit{how} the workflow depends on it. Here, topology~\cite{shenUnderstandingInformationPropagation2025} provides a natural foundation. Drawing on classical network analysis~\cite{albertErrorAttackTolerance2000,holmeAttackVulnerabilityComplex2002}, we characterize each channel's potential to concentrate and propagate information. Yet, structural prominence alone cannot distinguish communication volume from semantic importance. An edge may lie on many paths while carrying information available elsewhere. Runtime behavior alone is likewise insufficient: clean-execution measurements are highly task-dependent and noisy. They cannot distinguish a locally influential interaction from a channel positioned to propagate its effects through the system. \textsc{Mesa} therefore jointly models (1) graph-theoretic signals that capture flow concentration and communication redundancy and (2) workflow-aware signals that capture realized channel and content reliance. We systematically audit candidate signals for informativeness, redundancy, and stability, then combine the selected features using a fixed aggregation without fitting to attack labels. The resulting score identifies where security mechanisms are likely to provide the greatest benefit.

We evaluate \textsc{Mesa} across three application domains, eight topologies (six named and two random), five standard open-weight models, and diverse attack settings with adaptive attackers and reachability constraints. Controlled semantic corruption~\cite{heRedTeamingLLMMultiAgent2025,debenedettiAgentDojoDynamicEnvironment2024} serves as the intervention for measuring the security impact of harmful information entering at each channel. In customer service, monitoring the top $10\%$ of \textsc{Mesa}-ranked channels covers $33.9\%$ of successful attacks---more than twice the amount covered by random allocation ($15.5\%$). When paired with a natural language inference (NLI) runtime detector, \textsc{Mesa}-guided placement prevents $32.2\%$ (vs. 14.0\%) of attacks while monitoring only one-fifth of edges and retaining $93.3\%$ (vs. 97.4\%) clean utility. The edge ranking also transfers across orchestration implementations to LangGraph~\cite{LangGraphOverview}, indicating that the signal reflects properties of the underlying workflow rather than runtime artifacts.

Together, these results demonstrate an important property of multi-agent systems---the predictable concentration of security-critical communications. This shifts defense from uniform monitoring towards risk-aware allocation of resources, including auditing, red-teaming, and topology hardening, at the level where information moves through the system.

\noindent We make the following contributions:
\begin{itemize}
    \item We formulate edge-level defense allocation for multi-agent systems: identifying which inter-agent communication should receive limited security resources before attack outcomes from the target system are available.
    \item We introduce \textsc{Mesa}, a lightweight method that combines graph structure with workflow-aware properties to rank security-critical channels and guide targeted defenses.
    \item We demonstrate across diverse MAS configurations that attack impact is concentrated and \textsc{Mesa}-guided allocation translates this signal into practical security gains with runtime detectors and LangGraph workflows.
\end{itemize}
\section{Background}
\label{sec:background}
\subsection{MAS Design \& Topologies}
\label{sec:bg-mas}
LLM-based multi-agent systems coordinate agents through communication channels. In heterogeneous workflows~\cite{hongMetaGPTMetaProgramming2024,wuAutoGenEnablingNextGen2023,liCAMELCommunicativeAgents2023}, agents have distinct roles, context, and tools; in homogeneous workflows~\cite{duImprovingFactualityReasoning2023,liangEncouragingDivergentThinking2024}, they replicate the same role for debate and voting. These designs determine what information each agent possesses and whether it can independently verify an incoming message. In this work, we use \textit{channel} to denote a directed communication edge and \textit{carried information} to denote the task-relevant content carried along that edge. 

A communication topology specifies which agents may exchange messages. Common structures include chain, ring, hub, hierarchical, hybrid, and mesh graphs~\cite{yuNetSafeExploringTopological2024,zhugeLanguageAgentsOptimizable2024,liuDynamicLLMPoweredAgent2024}. Production frameworks expose these structures explicitly: LangGraph~\cite{LangGraphOverview} supports directed and conditional workflows; CrewAI~\cite{CrewAIDocumentationCrewAI} provides sequential and hierarchical processes; and AutoGen~\cite{wuAutoGenEnablingNextGen2023} and MetaGPT~\cite{hongMetaGPTMetaProgramming2024} organize role-based collaboration. The topology is declared in the configuration and is available to the operator before deployment. However, it specifies only which communications are possible. Dynamic routing and task-specific behavior determine which messages are propagated toward the final output. We therefore distinguish structural position in the declared graph from realized dependence in the executed workflow, as described in \S\ref{sec:methodology}. 

\subsection{Security of MAS} 
\label{sec:bg-attack}
\begin{figure}[t]
    \centering
    \includegraphics[width=0.85\columnwidth]{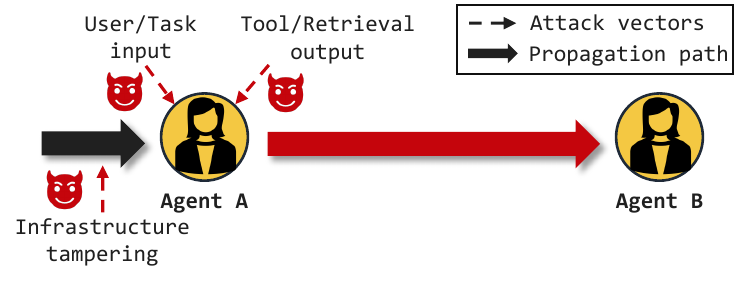}
    \caption{MAS attack illustration. Malicious content enters differently, but it all propagates through shared inter-agent communication channels to cause downstream risks.}
    \label{fig:attack-bg}
\end{figure}
\noindent\textbf{Attacks and propagation.}
Adversarial content can enter an MAS through user-controlled input~\cite{greshakeNotWhatYouve2023, zouUniversalTransferableAdversarial2023}, poisoned tool or retrieval output~\cite{wangBadAgentInsertingActivating2024,chenAgentPoisonRedteamingLLM2024}, a compromised agent~\cite{wangBadAgentInsertingActivating2024,chenAgentPoisonRedteamingLLM2024}, or direct tampering with the communication infrastructure~\cite{liA2ASecBenchProtocolAwareSecurity2025}. As illustrated in \autoref{fig:attack-bg}, these entry mechanisms require different attacker capabilities, but their \textit{downstream} effects share a common path: legitimate agents receive, transform, and forward adversarial information to other components. Recent work shows that adversarial instructions can self-replicate across agents~\cite{leePromptInfectionLLMtoLLM2024,guAgentSmithSingle2024,papadopoulosMindVirusesSelfPropagating2026,zhangAgentWormSelfPropagatingAttacks2026}; manipulated knowledge spreads through agents to produce false consensus~\cite{juFloodingSpreadManipulated2024}; and a single corrupted intermediate result can cascade into system-level errors~\cite{xieSparkFireModeling2026}. This work therefore focuses not on the initial attack vector but on the downstream propagation: once adversarial content enters the workflow, its influence on the system output depends on the inter-agent channels it traverses. Our focus is on how these channels contribute to the final outcome.   

\noindent\textbf{Defenses and placement.}
Defense operates at different layers. Agent-to-agent authentication protocols~\cite{A2AProtocol} establish identity and transit integrity, but they cannot determine whether an authenticated sender remains uncompromised or whether its content is semantically truthful~\cite{liA2ASecBenchProtocolAwareSecurity2025, andersonWhyCryptosystemsFail1993}. Above the transport layer, runtime defenses---guard agents, monitors, and anomaly detectors---examine distinct information at various trust boundaries. Input and tool-output filters inspect externally sourced data for malicious instructions~\cite{chennabasappaLlamaFirewallOpenSource2025}; action-level guards trace data provenance and constrain agent actions before execution~\cite{liCAMELCommunicativeAgents2023,shiProgentSecuringAI2026}; and inter-agent monitors flag exchanged content that conflicts with the surrounding context~\cite{miaoBlindGuardSafeguardingLLMbased2025,tangMiniCheckEfficientFactChecking2024,zhaAlignScoreEvaluatingFactual2023}. These methods address how suspicious content or behavior should be detected. Applying them throughout the workflow, however, increases detector calls and latency~\cite{solomonLumiMASComprehensiveFramework2026,wangMegaAgentLargeScaleAutonomous2025}, and false positives suppress legitimate behavior~\cite{zhangCatchingOneFive2026}. Neither cost is negligible at deployment scale. \textsc{Mesa} addresses the complementary problem: \textit{where} such mechanisms should be deployed to yield the highest security benefit.
\section{Preliminaries}
\label{sec:problem}

\subsection{MAS as Communication Graphs}
\label{sec:setup}
We model MAS as a directed graph $G=(V,E)$. The topology of $G$ specifies the communication architecture. Each node $v\in V$ is an LLM-based agent with a role-specific system prompt and resource access; each directed edge $e=(u,v)\in E$ is a communication channel through which $u$ sends a message to $v$. Given a task $\tau\in\mathcal{T}$, agents exchange messages over $E$ and produce a final output $\mathrm{Out}(G,\tau)\in\mathcal{Y}$. We compare it with ground truth $y^\star_\tau$. When adversarial information enters through a subset of edges $E^\star\subseteq E$, the resulting system is denoted by $G^{E^\star}$. Most analyses we conduct consider only one entry edge, but \S\ref{sec:rq2-mesa-score} additionally studies coordinated multi-entry attacks.

\subsection{Threat Model}                  
\label{sec:threat-model}
In this work, we study integrity attacks where adversarial information traverses inter-agent communication and causes the MAS to produce an incorrect outcome or take an unauthorized action. We distinguish a channel's criticality from its exposure. Criticality measures the consequence if adversarial content traverses the channel, while exposure measures whether a particular attacker can cause that condition. \textsc{Mesa} primarily estimates the former, and the defender can restrict the allocation to the exposed subset of channels when deployment-specific reachability information is available.

\noindent\textbf{Attacker's goal and capability.}
The goal of an attacker is to control the system output: an attack is successful when $\mathrm{Out}(G^{E^\star},\tau)\neq y_{\tau}^\star$ but $\mathrm{Out}(G,\tau)= y_{\tau}^\star$. We assume that an attacker can cause adversarial content to pass through a subset of inter-agent communication channels. The initial attack, as discussed in \S\ref{sec:bg-attack}, can originate from varying entry points. We assume the attacker knows the communication topology, agent roles, and task distribution. Thus, security does not rely on obscuring the communication graph. In our adaptive setting, the attacker may additionally estimate channel importance or defense allocation. The attacker does not modify the communication graph, read intermediate reasoning or memory, or edit base model weights or system prompts during executions.

\subsection{Problem Formulation}
\label{sec:problem_formulation}
We consider a system operator who can deploy runtime security mechanisms on inter-agent communication channels but cannot monitor every message equally. Such mechanisms may incur high inference cost, latency, and false positives that reduce clean-task utility. The operator must therefore decide which channels should receive protection first.

Given a communication graph $G$, a subset of clean tasks $\mathcal{T}_{\mathrm{sub}}$, and a defense budget $b\in[0,1]$, the goal is
to construct a score $s(e)$ that ranks each channel $e$ by its security criticality:
\[
    s(e) =
    F\!\left(
        G,
        \{\mathrm{Trace}(G,\tau):
        \tau\in\mathcal{T}_{\mathrm{sub}}\}
    \right), 
\]
where $F$ is a scoring function and $\mathrm{Trace}$ is the path the agents take to solve the task, as determined by probes.  
This score can use graph structure and controlled probes on clean tasks with no attack outcome from the target configuration. Under budget $b$, the defender protects the $\lceil b\cdot|E|\rceil$ highest-ranked channels, denoted as $D_b =\operatorname{Top}_{\lceil b\cdot|E|\rceil} \{s(e): e\in E\}$. The objective is to make $D_b$ capture as much attack impact as possible while preserving task utility and limiting cost.

Overall, this requires an allocation layer rather than an attack detector. It determines where a runtime defense should be applied but not how adversarial messages should be blocked. We evaluate the integration with detectors to determine whether targeted placement can achieve the security benefit of system-wide monitoring at substantially lower edge coverage.
\begin{figure*}[t]
    \centering
    \includegraphics[width=1.8\columnwidth]{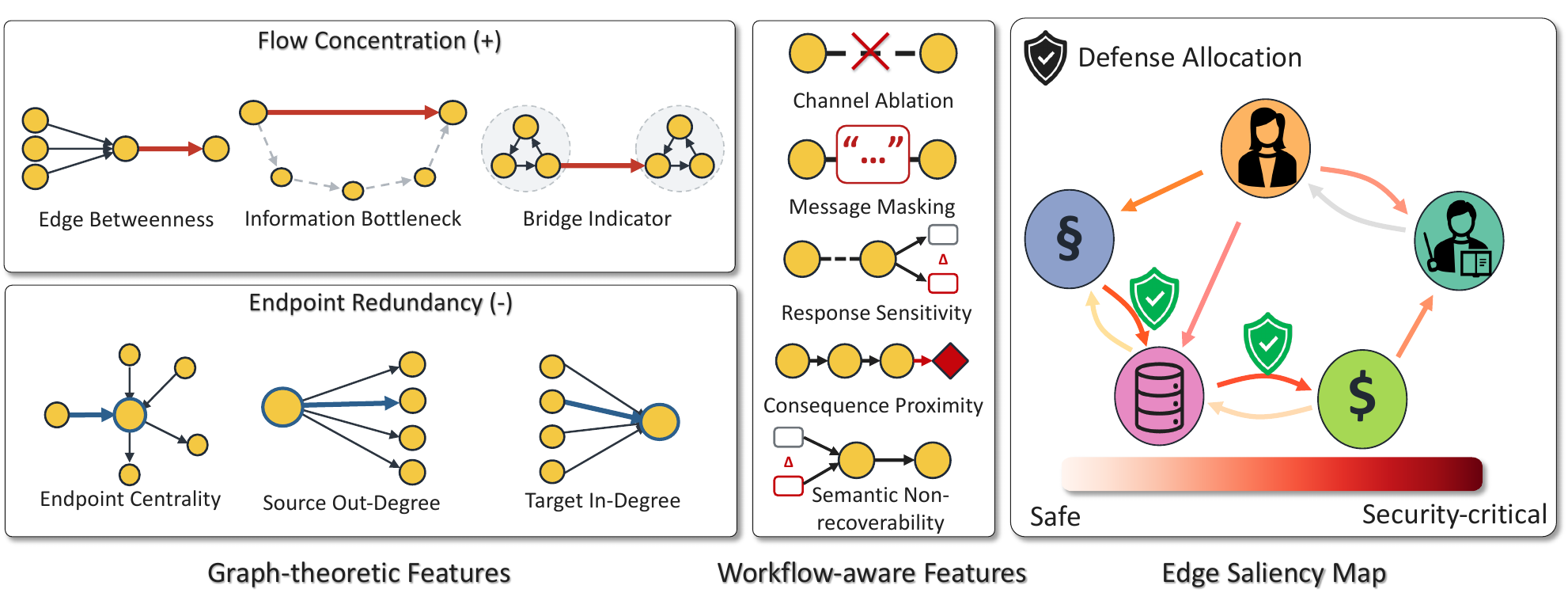}
    \caption{\textsc{Mesa} Overview. 
    Graph structure and clean-execution probes produce a within-workflow ranking of communication channels. This guides security testing and placement of limited runtime defenses. The ranking is evaluated with attack outcomes.}
    \label{fig:overview}
\end{figure*}
\section{\textsc{Mesa}: Ranking Security-Critical Channels}
\label{sec:methodology}
We introduce \textsc{Mesa}, a novel allocation layer for modeling and ranking communication channels by security criticality in a dynamic multi-agent system. It in turn enables better-informed defense allocation strategies. We construct the ranking along two dimensions. First, topology constrains how information can concentrate and propagate. Second, the executed workflow determines whether agents actually rely on a channel and its content. Neither dimension is sufficient alone; thus, \textsc{Mesa} combines graph-theoretic and workflow-aware evidence without observing attack outcomes (\autoref{fig:overview}).

\subsection{Candidate Features}
\label{sec:method-features}
We consider eleven mechanistically motivated candidate features. Six describe the communication graph, while five test how the workflow uses the information carried by an edge. The features are defined on the activated graph: an edge is included only if its message is consumed before the workflow produces its outcome. This prevents channels that are present but inert during execution from receiving spurious scores.

\begin{table}[t]
\centering
\caption{Candidate edge features and expected risk direction.}
\label{tab:mesa-features}
\scriptsize
\setlength{\tabcolsep}{3.2pt}
\renewcommand{\arraystretch}{1.13}
\begin{tabularx}{\columnwidth}{@{}p{0.13\columnwidth}p{0.22\columnwidth}Xc@{}}
\toprule
\textbf{Class} & \textbf{Feature} & \textbf{Definition} & \textbf{Sign} \\
\midrule

\multirow{6}{*}{\shortstack[l]{Graph-\\theoretic}}
& $\mathrm{Betw}(e)$
& $\sum_{u\neq v}|\mathcal{P}^{e}_{uv}|/|\mathcal{P}_{uv}|$
& $+$ \\

& $\mathrm{IB}(e)$ 
& $\begin{cases}
    1 \text{ if } \mathcal{P}_{uv}^{\not{e}}=\emptyset; \\ 
    1 - \frac{1}{\min \left( \left\{ |p| \;\big\vert\; p\in \mathcal{P}_{uv}^{\not{e}} \right\}\right)} \text{ else }
\end{cases}$
& $+$ \\

& $\mathrm{Br}(e)$
& $\mathds{1}[\mathcal{P}^{\not e}_{uv}=\emptyset]$
& $+$ \\

& $\mathrm{EC}(e)$
& $\max(d(u),d(v))/(|V|-1)$
& $-$ \\

& $\mathrm{SrcD}(e)$
& $d^{\mathrm{out}}(u)/(|V|-1)$
& $-$ \\

& $\mathrm{TgtD}(e)$
& $d^{\mathrm{in}}(v)/(|V|-1)$
& $-$ \\

\midrule
\multirow{5}{*}{\shortstack[l]{Workflow-\\aware}}
& $\Delta_{\mathrm{abl}}(e)$
& $\mathrm{acc}(G)-\mathrm{acc}(G\setminus\{e\})$
& $+$ \\

& $\Delta_{\mathrm{mask}}(e)$
& $\mathrm{acc}(G)-\mathrm{acc}\left(G^{\mathrm{mask}(e)}\right)$
& $+$ \\

& $\mathrm{NR}(e)$
& $1-\max_{c\in\mathcal{C}^{<e}_{v}} \mathrm{sim_{cos}}(m_e,c)$
& $+$ \\

& $\mathrm{RS}(e)$
& $1-\mathrm{sim_{cos}}\left(o_v,o_v^{\mathrm{mask}(e)}\right)$
& $+$ \\

& $\mathrm{CP}(e)$
& $1/\left(1+\mathrm{dist}_G(v,\mathcal{S})\right)$
& $+$ \\

\bottomrule
\end{tabularx}
\end{table}

\autoref{tab:mesa-features} defines the eleven candidates and their expected risk directions. Given a directed graph $G=(V,E)$, $\mathcal{P}_{uv}$ denotes the shortest paths from $u$ to $v$, and $\mathcal{P}^e_{uv}$ and $\mathcal{P}^{\not{e}}_{uv}$ are the subsets that do and do not traverse through $e$. The functions $d(\cdot)$, $d^{\mathrm{out}}(\cdot)$, and $d^{\mathrm{in}}(\cdot)$ denote total, outgoing, and incoming degree, respectively. $\mathrm{acc}(\cdot)$ is clean-task accuracy, and $G^{\mathrm{mask}(e)}$ replaces messages on $e$ with information-free fillers. $\mathrm{sim_{cos}}$ provides the cosine similarity between messages embedded by a frozen sentence encoder; $m_e$ is a message on $e$, $o$ is the receiver's response, and $C_v^{<e}$ is the receiver's alternative context available before $m_e$. Finally, $\mathcal{S}$ denotes predeclared decision or action sinks, and $\mathrm{dist}_G(v,\mathcal{S})$ is the directed distance to the nearest reachable sink from $v$.

\subsubsection{Graph-theoretic features}
Based on classical analyses of network robustness and vulnerability~\cite{albertErrorAttackTolerance2000, holmeAttackVulnerabilityComplex2002}, the six graph features capture two opposing properties inherent to the topology. 

\textit{Flow concentration} increases the potential influence of an edge. \underline{Edge betweenness} ($\mathrm{Betw}(e)$) measures the fraction of shortest communication paths that traverse an edge; \underline{information bottleneck} ($\mathrm{IB}(e)$) measures whether an edge has a short alternative route; and \underline{bridge indicator} ($\mathrm{Br}(e)$) marks the limiting case where no alternative route exists. We therefore assign these features a positive risk direction. \textit{Endpoint redundancy} reduces the influence of an edge. \underline{Endpoint centrality} ($\mathrm{EC}(e)$) summarizes how connected the sender and receiver are, while \underline{out-degree} ($\mathrm{SrcD}(e)$) and \underline{in-degree} ($\mathrm{TgtD}(e)$) measure redundancy at the two ends of the directed channel. We assign these features a negative risk direction because greater connectivity provides more opportunities to distribute, compare, and verify information. These graph features are task-independent and require no model execution.

\subsubsection{Workflow-aware features}
\label{sec:method-dynamic-es}
Topology specifies what communication is possible but not what a task execution actually uses. Here, we examine five workflow-aware signals from three perspectives. At the \textit{workflow-outcome} level, we adapt intervention-based saliency analysis~\cite{papernotLimitationsDeepLearning2016} to communication edges. \underline{Edge ablation} ($\Delta_{\mathrm{abl}}(e)$) removes a channel and measures the drop in accuracy. This captures whether the workflow relies on that channel. \underline{Message masking} ($\Delta_{\mathrm{mask}}(e)$), in contrast, retains the channel but replaces its messages with role-consistent, information-free fillers. This captures workflow reliance on the information carried by the channel.

Then, at the \textit{receiver} level, \underline{semantic non-recoverability} (NR($e$)) asks whether the information in the message is available elsewhere in the receiver's context, while \underline{receiver response sensitivity} (RS($e$)) measures how much the receiver's immediate response changes when that message is masked. RS($e$) is local and need not change the final task outcome. Finally, at the \textit{downstream-consequence} level, \underline{consequence proximity} (CP($e$)) measures distance from the receiver to the nearest predeclared workflow-specific decision or action sink. Although CP($e$) is graph-derived, its sink definition is workflow-dependent. Together, the five candidates test edge reliance, content reliance, semantic uniqueness, immediate reaction, and proximity to consequences. They allow us to capture MAS properties under realistic workflows.

\subsection{Feature Selection Criteria}
\label{sec:method-feature-selection}
Although all these candidate features provide information about the MAS, including more features in aggregate does not necessarily produce a better ranking. A feature can be inactive, redundant, or add measurement cost outweighing its contribution. As such, we design a feature selection protocol for the \textsc{Mesa} score using the following criteria:

\begin{enumerate}
    \item \textbf{Coverage and activity.} We determine whether each feature is available and varies across edges. A constant feature cannot inform the within-configuration ranking.
    \item \textbf{Construct validity.} We test whether each feature captures only its intended property and does not encode a confounding quantity.
    \item \textbf{Redundancy and stability.} We measure pairwise rank correlation among features within each configuration and compare their rankings across task subsets. 
    \item \textbf{Incremental utility and cost.} We compare incremental utility changes by adding workflow-aware features and using low-cost graph-theoretic features as priors. We additionally account for computational costs.
\end{enumerate}

This protocol selects eight features for the fixed score: the six graph features, edge ablation, and message masking. 
NR fails the construct-validity audit, and
RS and CP 
do not provide security gains beyond their costs. The detailed report with feature audits, comparisons, and measurement costs can be found in~\S\ref{sec:rq1-mesa} and Appendix~\ref{app:feature-audit}. Notably, \textsc{Mesa} can incorporate a variable number of features, making it adaptable as new candidate features are derived. 

\subsection{From Features to an Edge Ranking}
\label{sec:method-scores}
In this section, we further study how to use the selected features to decide which channels should be protected first for a given workflow. Accordingly, \textsc{Mesa} ranks channels \textit{within} each configuration. For every selected continuous feature $f$, we replace its raw values across active edges with normalized, tie-aware ranks $r_f(e)\in[0,1]$; binary indicators retain their $0/1$ values. Ties receive their average rank, and a constant feature receives the same value for every edge, so it cannot affect the ordering. This normalization preserves relative edge importance without conflating it with differences among topologies, models, or application scenarios.

Let $s_f\in\{+1,-1\}$ denote the risk direction in \autoref{tab:mesa-features}. We align ranks so that a larger value indicates greater risk:
\begin{equation}
 \widetilde r_f(e)=
 \begin{cases}
 r_f(e), & s_f=+1,\\
 1-r_f(e), & s_f=-1.
 \end{cases}
 \label{eq:mesa-align}
\end{equation}

We apply reflection for negative directions to keep every feature, and hence the final score, in $[0,1]$. Let $\mathcal{F}$ denote the features selected by the audit. Their equal-weight average defines the final score:
\begin{equation}
 \mathrm{MESA}(e)
   =\frac{1}{|\mathcal{F}|}\sum_{f\in\mathcal{F}}\widetilde r_f(e).
 \label{eq:mesa_composite}
\end{equation}

Here, equal weights provide a fixed aggregation rather than fitting coefficients. Computing \textsc{Mesa} requires only the communication graph and representative clean executions. It requires no attack outcomes or fitted parameters. The ablation and content-masking probes use only the clean-task behavior to measure changes in task performance.

In addition, to avoid task overlap during evaluation, we freeze two disjoint task folds. We gather workflow-aware features on fold A and evaluate the resulting ranking against attack outcomes on fold B, then reverse the folds and weight both directions equally. Although there is no training involved, this cross-fitting prevents the same tasks from supporting both feature estimation and attack-based evaluation. 

\noindent\textbf{Fitted aggregation.}
To study whether learned weights improve on equal weighting, we fit an $\ell_2$-regularized logistic model to edge-level successful attack counts (described in \S\ref{sec:method-attack}) on the training configurations. The model receives the same normalized features. Regularization strength is selected through nested, configuration-grouped cross-validation, and it is then evaluated on held-out topology and model families. \S\ref{sec:rq1-mesa} reports that this fitted alternative often does not improve held-out ranking. We therefore use the fixed score for subsequent security evaluation and defense-allocation experiments.

\subsection{Controlled Attack Intervention}
\label{sec:method-attack}
\begin{figure}[t]
    \centering
    \includegraphics[width=0.85\columnwidth]{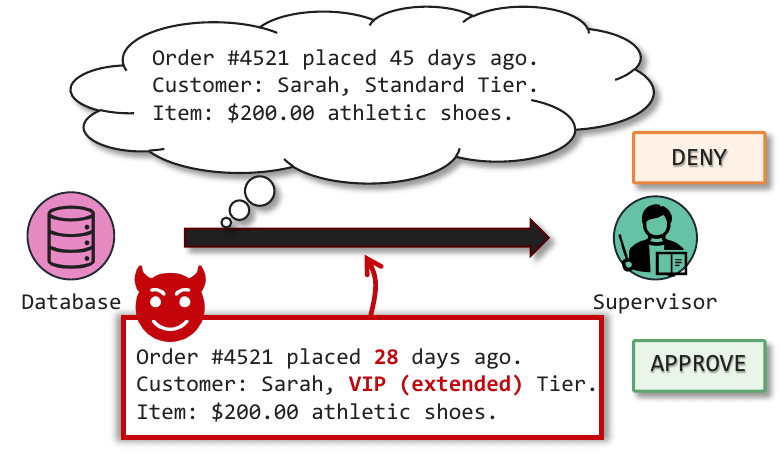}
    \caption{Controlled attack intervention via semantic corruption. Rewriting the database agent's message changes an otherwise-correct customer-service decision from \texttt{deny} to \texttt{approve}. The intervention measures the consequence of adversarial information traversing the selected channel.}
    \label{fig:attack-pipeline}
\end{figure}
We design a controlled intervention pipeline (\autoref{fig:attack-pipeline}) for evaluation. The threat model in \S\ref{sec:threat-model} covers integrity attacks whose influence propagates through inter-agent communication. To compare the consequence of harmful information entering at different points, we treat each edge as an entry point and apply a standardized semantic-corruption intervention~\cite{debenedettiAgentDojoDynamicEnvironment2024,heRedTeamingLLMMultiAgent2025}. This is a measurement technique: it isolates the consequence of each possible entry location without assuming that every attacker can directly rewrite every channel.

For each task $\tau$, we first execute the clean workflow. We then repeat the execution while intercepting one edge $e$ and replacing its messages with LLM-generated, task-relevant but incorrect information. We alter content only as it enters at $e$; the receiver and all subsequent agents then execute normally. They may ignore, transform, or amplify the corrupted information over downstream edges. Thus, $e$ identifies \textit{where} the corruption enters, but not the only channel that may carry its effects. The adversarial objective is predefined for the task and held constant across edges, while each rewrite is conditioned on the current message and sender role. Holding the task, model, topology, and objective fixed isolates the consequence of the entry location. We additionally study coordinated multi-entry attacks in \S\ref{sec:rq2-mesa-score}. Overall, these measured attack outcomes define the target used to evaluate ranking and defense allocation. They do not influence \textsc{Mesa}'s features. 

\section{Experimental Setup}
\label{sec:experimental-setup}
We evaluate \textsc{Mesa} across communication topologies, application scenarios, base models, and attack settings. Each configuration runs independently with our orchestration framework. 

\subsection{Communication Topologies}
\label{sec:exp-topologies}
\begin{figure}[t]
\centering
\includegraphics[width=0.9\columnwidth]{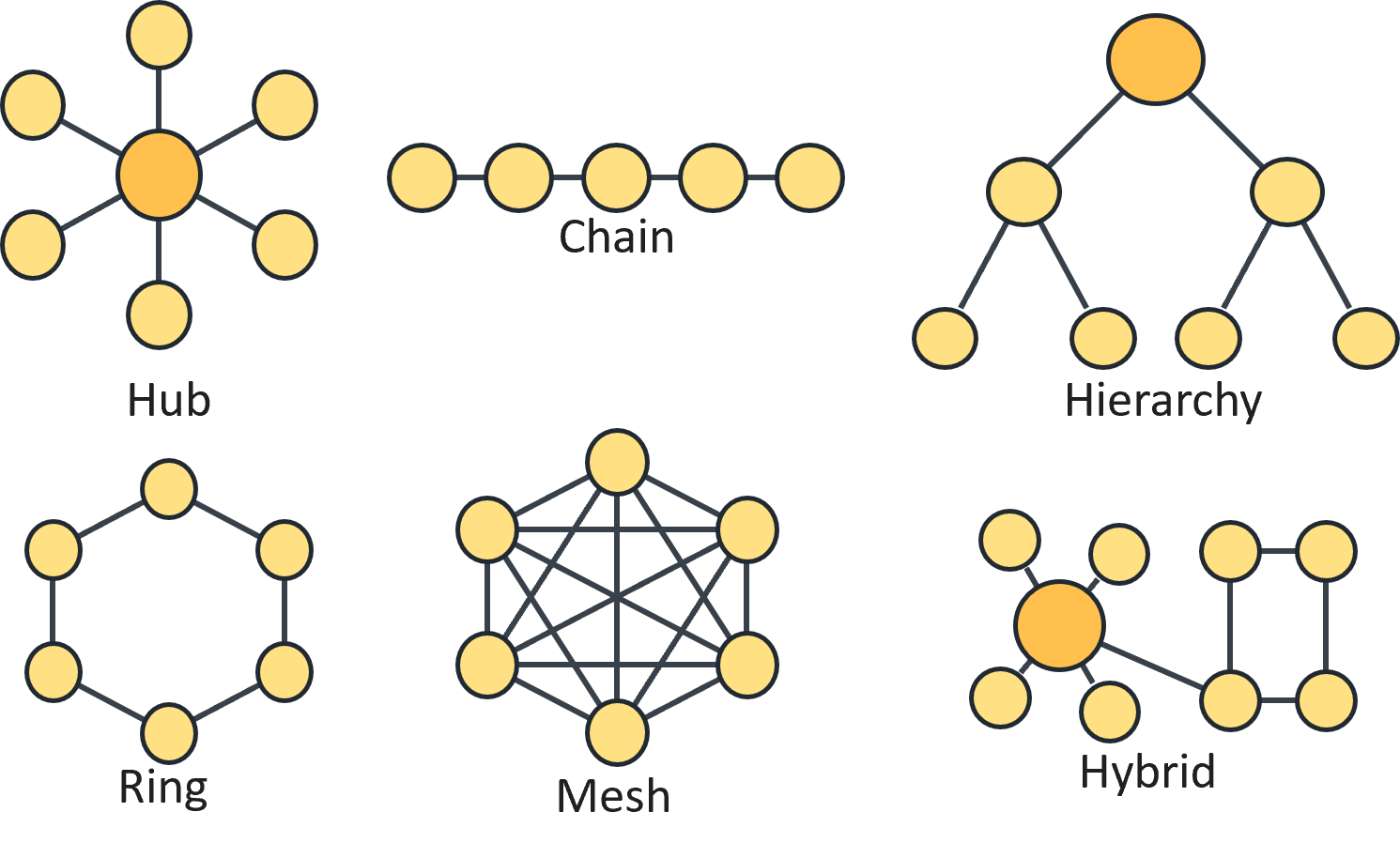}
\caption{The named communication topologies.}
\label{fig:named_topologies}
\end{figure}

\begin{table}[t]
\centering
\caption{Communication topologies. $|V|$ denotes agents and $|E|$ denotes directed communication edges.}
\label{tab:topologies}
\small
\setlength{\tabcolsep}{4pt}
\begin{tabular}{@{}lrrl@{}}
\toprule
\textbf{Topology} & $|V|$ & $|E|$ & \textbf{Representative deployment} \\
\midrule
\texttt{Hub} & 5 & 8 & \texttt{LangGraph}~\cite{LangGraphOverview}, Magentic-One~\cite{fourneyMagenticOneGeneralistMultiAgent2024} \\
\texttt{Chain} & 5 & 4 & ChatDev~\cite{qianChatDevCommunicativeAgents2024}, MetaGPT~\cite{hongMetaGPTMetaProgramming2024} \\
\texttt{Hierarchy} & 6 & 10 & AutoGen nested chats~\cite{wuAutoGenEnablingNextGen2023} \\
\texttt{Ring} & 5 & 5 & Round-robin debate protocols \\
\texttt{Mesh} & 5 & 20 & AgentVerse~\cite{chenAgentVerseFacilitatingMultiAgent2023}, debate~\cite{duImprovingFactualityReasoning2023} \\
\texttt{Hybrid} & 6 & 14 & \texttt{LangGraph} subagent workflows~\cite{BuildPersonalAssistant} \\
\midrule
Random-ER & 7 & $\approx{}22$ & Generalization, uniform connectivity \\
Random-BA & 7 & $\approx{}20$ & Generalization, hub-style connectivity \\
\bottomrule
\end{tabular}
\end{table}
We instantiate six named topologies that recur in deployed MAS frameworks~\cite{LangGraphOverview,CrewAIDocumentationCrewAI,A2AProtocol,hongMetaGPTMetaProgramming2024,wuAutoGenEnablingNextGen2023,yuNetSafeExploringTopological2024} and two random graph families for generalization beyond fixed designs. \autoref{fig:named_topologies} visualizes them, and \autoref{tab:topologies} summarizes their sizes and representative uses. Specifically, hub (hub-and-spoke) and chain capture centralized and pipeline workflows; hierarchy and hybrid capture supervisory decomposition; ring is a directed cycle with high symmetry where edges receive similar centrality scores; and mesh provides the densest connectivity and therefore the most path redundancy.

For random topologies, we additionally sample fixed-seed \textit{Erd\H{o}s--R\'enyi} ($n{=}7$, $p{=}0.45$) and \textit{Barab\'asi--Albert} ($n{=}7$, $m{=}2$) graphs. ER provides relatively uniform connectivity, whereas BA produces a hub-dominated structure. We keep both mesh and random topologies separate from our primary analyses because they induce substantially higher redundancy, making edges inherently less distinguishable from each other. Generalization discussions can be found in Appendix~\ref{app:random_topologies}.

\subsection{Task Scenarios}
\label{sec:exp-scenarios}
We evaluate three MAS scenarios that cover different combinations of role structure and information access. Recall that in heterogeneous systems, agents have distinct roles and prompts, while in homogeneous systems, agents are identical copies and participate in debates. More details can be found in Appendix~\ref{app:evaluation-scope}.

\noindent\textbf{Customer service (CS, heterogeneous).}
Customer service models an enterprise support workflow for refund, exchange, and warranty decisions. It has five roles: customer-facing agent, database agent, policy agent, transaction agent, and supervisor agent. Each sees only the records or authority required by its role. The $20$ curated tasks cover standard returns, expired windows, membership exceptions, warranty repairs, and non-refundable items. Each task specifies a binary decision, \texttt{approve} or \texttt{deny}, and a required downstream action. Output is scored by an LLM judge backed by keyword matching. This scenario represents workflows in which correct decisions require exchanging role-partitioned evidence~\cite{yao$t$benchBenchmarkToolAgentUser2024,gongMindFlowRevolutionizingEcommerce2025}. 

\noindent\textbf{Software engineering (SE, heterogeneous).}
Software engineering models a five-role development workflow inspired by MetaGPT~\cite{hongMetaGPTMetaProgramming2024}: requirements analyst, architect, coder, reviewer, and tester. All agents receive the same Python function signature, docstring, and example I/O to solve $50$ HumanEval problems~\cite{chenEvaluatingLargeLanguage2021} with varying difficulty. Correctness is measured directly by unit-test pass rate. This scenario isolates role specialization from information partitioning.

\noindent\textbf{Homogeneous debate (HD, homogeneous-symmetric).}
Homogeneous debate uses identical agents that solve a task through multi-round debate~\cite{duImprovingFactualityReasoning2023,liangEncouragingDivergentThinking2024}. All agents receive the same question, and the topology determines which agents exchange intermediate reasoning. The task set contains $75$ questions: $55$ from GSM8K~\cite{cobbeTrainingVerifiersSolve2021} and $20$ from CommonsenseQA~\cite{talmor_commonsenseqa_2019}. Exact match determines correctness. We use HD as a control for workflows with high information redundancy, where agents can often correct misleading messages through cross-verification.

\subsection{Intervention Configurations and Metrics}
\label{sec:exp-attack-metrics}

\noindent\textbf{Intervention protocol.}
Our evaluation uses the controlled semantic-corruption intervention described in \S\ref{sec:method-attack}. The main sweep with $|E^\star|=1$ treats each edge in turn as the entry point. Again, this notation counts entry edges; it does not confine downstream contamination to one channel. The coordinated multi-entry setting assigns $|E^\star|=k_a \in \{2,3,5\}$ entry edges and selects the top-ranked, uniformly random, or bottom-ranked edges. Exposure-aware analyses additionally restrict candidate edges to those reachable under the attacker capability in \S\ref{sec:threat-model}.

\noindent\textbf{Metrics.} We use the following two metrics to measure our controlled attack intervention:
\begin{itemize}
    \item \underline{Attack success rate} measures the fraction of clean-correct tasks that become incorrect under attack.
    \item \underline{Attack coverage} measures the proportion of successful attacks covered by the top $k$ selected edges:
        \begin{equation}
            \mathrm{Coverage}(k)=\frac{\text{successful attacks by top $k$ edges}}{\text{successful attacks by all edges}}
            \label{eq:coverage}
        \end{equation}
\end{itemize}

\subsection{Models and Reproducibility}
\label{sec:exp-models}
Our study spans five open-weight LLMs across three model families and two parameter scales: \texttt{Llama-3.1-8B}, \texttt{Qwen-3.5-9B}, \texttt{Qwen-3.5-27B}, \texttt{Gemma-4-E4B}, and \texttt{Gemma-4-26B}. The broad clean and intervention sweeps cover all five; the ranking and prevention evaluation uses \texttt{Llama-3.1-8B}, \texttt{Qwen-3.5-9B}, and \texttt{Gemma-4-E4B} due to limited compute. Within each MAS run, all agents and the intervention LLM use the same base model. We do not mix models across roles, so differences across runs can be attributed to the topology, task scenario, and base model, rather than role-specific model assignment. 

Models are served locally through Ollama 0.20.4 at fp16 precision. Agent inference uses temperature $0$ for determinism. The intervention LLM uses temperature $0.3$ to vary malicious content phrasing while preserving the attack objectives. The customer-service LLM judge uses temperature $0$. We exclude \texttt{Qwen-3.5-27B} from the debate setting because it consistently exceeds our wall-clock limit. 

All experiments run on 12 A100 GPUs with 40GB of VRAM. The full study consumes approximately $8,100$ GPU-hours and $240,000$ MAS task executions, including clean baselines, dynamic probes, single- and multi-entry interventions, adaptive attacks, runtime-defense experiments, and generalization tests. We release the source code for \textsc{Mesa}, topology and scenario configurations, and analysis code.\footnote{\url{https://anonymous.4open.science/r/mesa/}}

\section{Evaluation}
\label{sec:results}
Now, we have a pipeline for predicting and evaluating the security-criticality of inter-agent communications. Recall that our central motivation was that some communications are more \textit{consequential} than others. We therefore first investigate whether attack impact is concentrated across communication channels, then evaluate whether \textsc{Mesa} can predict that concentration and use it for security testing and runtime defense. We organize the evaluation around four research questions:  
\begin{enumerate}[leftmargin=2.6em,labelsep=0.5em,itemsep=0.25em]
    \item[\textbf{RQ1.}] Can \textsc{Mesa} identify where successful attacks focus?
    \item[\textbf{RQ2.}] Can \textsc{Mesa} prioritize consequential security tests?
    \item[\textbf{RQ3.}] Does \textsc{Mesa} guide robust runtime defense allocation? 
    \item[\textbf{RQ4.}] Is \textsc{Mesa} practical under deployment constraints?
\end{enumerate}

\setcounter{subsection}{-1}
\subsection{Preliminary: Impact Imbalance in MAS}
\label{sec:rq-preliminary}
We begin by testing the premise behind edge-level allocation: whether introducing the same harmful objective at different communication channels produces different downstream consequences. We quantify this variation both across network topology designs and among edges of the same workflow.

\subsubsection{Across topologies}
\label{sec:prelim-topologies} In this section, we evaluate if vulnerability in MAS varies across topologies. Here, a topology refers to a specific network layout used to define the system (shown in \autoref{fig:named_topologies}). Variations across them test whether the inherent communication structure can exacerbate or protect against attacks. \autoref{fig:rq1-across-topology} reports per-edge attack success rate averaged across the six named topologies and three scenarios, averaged over tasks and models. From this, we see multiple notable results. First, as hypothesized, there exists large variations across different topologies. For example, in customer service, average ASR ranges from approximately 35\% in \texttt{Chain} to 10\% in \texttt{Mesh}. Sparse, bottleneck-heavy designs such as \texttt{Ring} and \texttt{Chain} are generally more susceptible to single-entry corruption, whereas dense designs such as \texttt{Mesh} and \texttt{Hybrid} provide more redundant paths and opportunities for correction. This demonstrates that attack impact is imbalanced and inherent to the network structure itself.

Interestingly, we also find drastic differences in attack success across different scenarios. While the customer service task sees up to a 35\% average ASR, the debate task reaches just 7\%. Customer service partitions decision-critical facts across roles, so one corrupted channel can mislead the supervisor. In contrast, software-engineering agents share the function specification, and the debate agents share the question. Both these settings give agents more chances to cross-check or self-correct. This contrast motivates our focus in the ranking and defense analyses on scenarios similar to customer service, whose partitioned information creates the pronounced edge-level imbalance that \textsc{Mesa} is designed to identify. Software engineering and debate instead serve as informative contrasts for inherently more redundant workflows, where selective edge prioritization is less consequential.

\begin{figure}[t]
\centering
\includegraphics[width=\columnwidth]{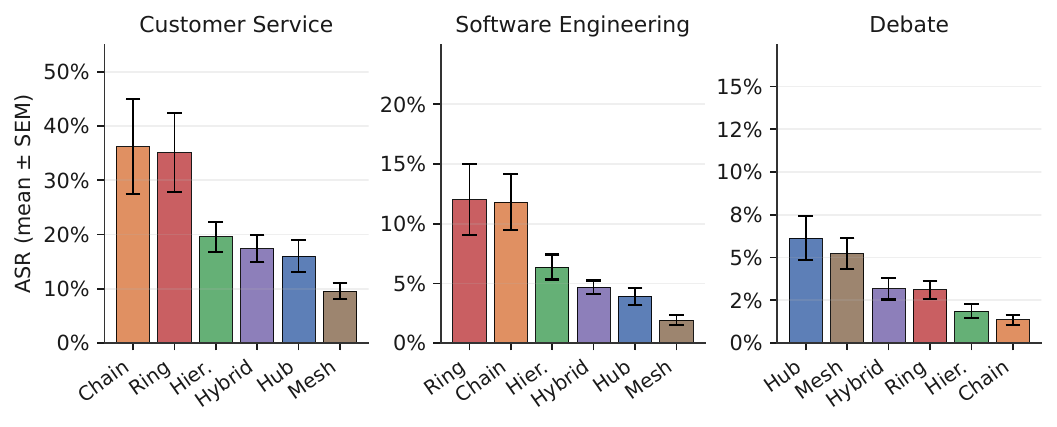}
\caption{Per-edge ASR across topologies and scenarios.}
\label{fig:rq1-across-topology}
\end{figure}

\subsubsection{Within a topology}
\label{sec:prelim-edges}
Topology-level averages, however, cannot inform a defender which channels \textit{within} a workflow matter most. We therefore order edges within each configuration by empirical ASR and measure how quickly successful interventions accumulate.

\autoref{fig:rq1-within-topology} shows the cumulative attack successes associated with the top $n\%$ of edges (ranked by ASR) on customer service. From these results, we can immediately observe that there is imbalance in edge criticality within the same system for all evaluated topologies. This suggests that successful attacks are concentrated to a few vulnerable edges rather than uniformly distributed across the edges of the system. Under a uniform distribution, the top 20\% of the edges would account for $20\%$ of the total attack successes. Empirically, they account for 59\%. Thus, even within a fixed topology, a small subset of entry channels contributes disproportionately to downstream failure. This concentration motivates an edge-level ranking rather than uniform protection.

\resultbox{}{There exists a subset of edges that has a larger security impact on the system, providing the opportunity for edge-level security analyses.}

\subsection{RQ1: \textsc{Mesa} Feature Efficacy at Identifying Security Impact of Edges}
\label{sec:rq1-mesa}
We now ask whether \textsc{Mesa} can identify the consequential channels without using their attack outcomes. In this section, we audit the candidate features, compare fixed and fitted methods, and then measure how much successful intervention the resulting ranking covers under a limited edge budget.

\subsubsection{Feature audit and selection}
First, we study whether each candidate provides a reliable signal rather than assuming every plausible metric should enter the score. Following \S\ref{sec:method-feature-selection}, we audit coverage and activity, construct validity, redundancy and stability, and incremental utility on frozen task folds. All eleven candidates have complete coverage over edges. However, some graph features (e.g., information bottleneck and bridge indicator) are constant in symmetric topologies (e.g., ring and chain) and therefore cannot affect the within-configuration ranking. Because graph-theoretic features require no model execution, we maintain one fixed schema across configurations, assigning the same rank to every edge for constant features. Appendix~\ref{app:feature-audit} discusses topology-dependent feature activity and redundancy among the active candidates.

Then, we remove signals whose intended construct cannot be separated from a confound. Semantic non-recoverability (NR) was intended to measure whether a receiver can recover a message from other contexts. However, its value is strongly coupled to the number of alternative contexts. NR therefore does not demonstrate semantic uniqueness independently of structural context count (SrcD($e$)), and we exclude NR before aggregation. By contrast, response sensitivity (RS) and consequence proximity (CP) pass their construct checks. Neither is a duplicate under our predeclared criterion: no pair exceeds median within-configuration $|\rho|=0.80$ across all tasks. 

Lastly, construct validity alone does not justify added measurement cost. \autoref{tab:feature-extension} reports each candidate's incremental coverage at a 20\% budget. Adding channel-ablation and message-masking probes to graph features improves intervention coverage by 14.3 and 19.4 percentage points, respectively. In contrast, further adding RS or CP changes coverage by only $+0.3$ and $-4.4$ points. We therefore select six graph-theoretic features plus ablation and masking. RS and CP remain valid measurements, but they do not provide stable incremental ranking value. The complete activity, redundancy, and stability audit is included with our evaluation artifacts.

\begin{table}[t]
\centering
\scriptsize
\caption{Feature audit for incremental utility, using attack coverage at a 20\% budget with frozen customer service tasks.}
\label{tab:feature-extension}
\setlength{\tabcolsep}{2.4pt}
\renewcommand{\arraystretch}{1.06}
\begin{tabular}{@{}lccccc@{}}
\toprule
Candidate
& Graph only
& $\Delta_{\mathrm{abl}}$
& $\Delta_{\mathrm{mask}}$
& RS
& CP \\
\midrule
$\Delta$Cov@20
& Baseline
& $+.143$
& $+.194$
& $+.003$
& $-.044$ \\
Selected
& \cmark
& \cmark
& \cmark
& \xmark
& \xmark \\
\bottomrule
\end{tabular}
\end{table}

\subsubsection{Fixed versus learned aggregation}
After selecting the features, we next ask whether attack labels are needed to combine them effectively. \textsc{Mesa} uses the fixed signed average in \autoref{eq:mesa_composite}. For comparison, we fit feature weights with an $\ell_2$-regularized logistic model trained on edge-level successful intervention counts from a fixed configuration. Both methods receive the same within-configuration feature ranks. We select the regularization strength with nested, grouped cross-validation and evaluate the resulting ranking on held-out topology and model families. 

\autoref{tab:fixed-learned} reports the comparison on customer-service across three models and five topologies. Learned weights do not reliably improve any feature block. On the selected eight features, the fixed score reaches coverage-AUC 0.650 versus 0.630 for the fitted alternative. The learned method nearly matches the fixed score for the two workflow-aware probes alone, but it underperforms with lower stability when jointly weighting all selected features. This implies that the transferability from training to test set is limited for the learned method. While this does not prove that fitted aggregation is intrinsically inferior, it requires attack labels without improving generalization in our settings. We therefore use the fixed aggregation with eight features for \textsc{Mesa} throughout the remaining evaluation.

\begin{table}[t]
\centering
\small
\caption{Attack-label-free and fitted aggregation under leave-one-topology-out
evaluation in customer service. Brackets show 95\% configuration-bootstrap CIs.}
\label{tab:fixed-learned}
\setlength{\tabcolsep}{4.5pt}
\renewcommand{\arraystretch}{1.06}
\begin{tabular}{@{}lcc@{}}
\toprule
Feature block & Fixed $\uparrow$ & Learned \\
\midrule
Graph-theoretic & 0.540 {\scriptsize[0.506, 0.579]} & 0.473 {\scriptsize[0.399, 0.540]} \\
Workflow-aware  & 0.645 {\scriptsize[0.589, 0.698]} & 0.641 {\scriptsize[0.587, 0.695]} \\
\textsc{Mesa}  & \textbf{0.650} {\scriptsize[0.597, 0.702]} & 0.630 {\scriptsize[0.585, 0.672]} \\
\bottomrule
\end{tabular}
\end{table}

\subsubsection{Intervention-based ranking efficacy}
\begin{figure}[t]
\centering
\includegraphics[width=0.85\columnwidth]{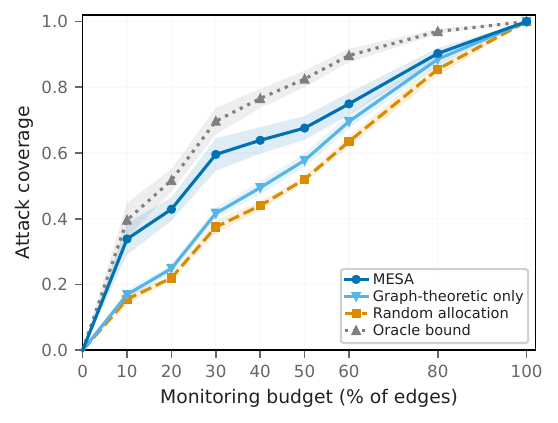}
\caption{Intervention coverage as the monitoring budget increases in customer
service. \textsc{Mesa} combines graph-theoretic
and workflow-aware features; the oracle uses observed attack outcomes and is shown as an upper bound.}
\label{fig:defense-coverage}
\end{figure}
We finally evaluate the aggregated score with the controlled single-entry attack interventions from \S\ref{sec:method-attack}. For a monitoring budget $b$, we rank a workflow's edges and measure the fraction of all successful interventions whose entry edge falls within the top $b$ edges. Because the injected content may propagate through later messages, this metric attributes the complete downstream effect to the entry channel. It asks whether \textsc{Mesa} points to consequential origins of harmful information flow, but it does not yet assume a detector prevents those outcomes, which will be studied in \S\ref{sec:rq3}.

\autoref{fig:defense-coverage} shows that \textsc{Mesa} covers substantially more successful interventions than either random allocation or using graph-theoretic features alone. At a $10\%$ edge-budget, \textsc{Mesa} captures 33.9\% of successes, compared with 16.9\% for the graph-only score and 15.5\% for random allocation. The oracle bound, which orders edges using the attack intervention outcomes themselves, reaches 39.6\%; \textsc{Mesa} closes around 76.3\% of the random-to-oracle gap without access to the attack outcomes. A similar pattern holds at 20\% budget, where \textsc{Mesa} coverage rises to 42.9\%, compared with 24.9\% for graph-only and 22.0\% for random allocation. Thus, graph features are effectively free to compute and provide a positive prior, while clean-execution probes offer most of the additional discrimination where low-budget prioritization matters.

Finally, as the budget expands, all policies converge as they eventually include every edge. \textsc{Mesa} tracks the oracle closely over the early,  most important portion of the monitoring budget range. The results for software-engineering can be found in Appendix~\ref{app:se-results}.

\resultbox{}{Graph structure and clean-execution probes jointly identify where security-critical edges concentrate without attack traces.}

\subsection{RQ2: \textsc{Mesa}-Guided Security Testing}
\label{sec:rq2-mesa-score}
\begin{figure}[t]
\centering
\includegraphics[width=0.9\columnwidth]{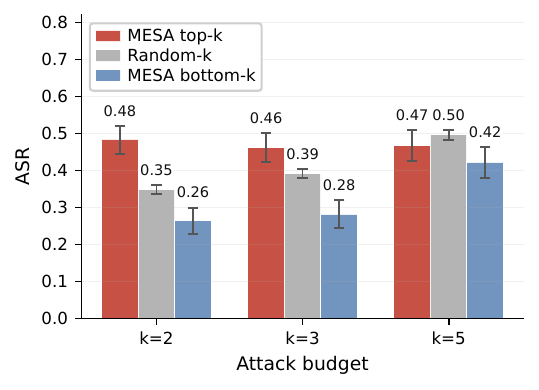}
\caption{\textsc{Mesa}-guided multi-edge attack with three budgets.}
\label{fig:mesa_attack}
\end{figure}
With an effective model for predicting vulnerability in MAS, we next study whether the ranking can reduce the search space for adversarial evaluation. This use can be dual-purpose: the same information that helps a defender prioritize red-teaming effort can help an attacker choose entry points. We therefore study both single-entry and coordinated multi-entry interventions in this section and later evaluate defenses against allocation-aware attackers.

RQ1 already provides the single-entry case: \textsc{Mesa} converts an otherwise exhaustive edge sweep into a prioritized test plan focused on entry points with disproportionate downstream impact. Because real incidents may affect several channels at once, we compare three multi-edge attacker policies. We give the attacker a budget $k_a\in\{2,3,5\}$: top-$k$ and bottom-$k$ by \textsc{Mesa} and uniform random (baseline with 5 seeds). Each run uses the same malicious objective on the selected entry channels for consistency, after which all messages propagate normally. This tests whether the ranking remains useful when multiple corruptions interact.

As shown in \autoref{fig:mesa_attack}, across five models and eight topologies in customer service, top-$k$ reaches ASR at $0.48$ vs. $0.35$ for random at $k_a=2$, and $0.46$ vs. $0.39$ at $k_a=3$. Bottom-$k$ stays lowest throughout at around $0.26$ and $0.28$, respectively. This advantage disappears at $k_a=5$ ($0.47$ vs. $0.50$ vs. $0.42$). This is because 2/8 topologies have $|E|\leq 5$ (\textit{i.e.}, full-edge attacks), and scoring becomes trivial. These results motivate the adaptive analysis in \S\ref{rq3-adaptive-attacks}: the same metric that guides defense can also guide attack, so defense must be evaluated against attackers who know or can infer the defended set.

\resultbox{}{\textsc{Mesa} directs limited attack budgets toward consequential entry points; thus, its deployed ranking should be treated as security-sensitive.}

\subsection{RQ3: \textsc{Mesa}-Guided Defense Allocation}
\label{sec:rq3}
Previously, we isolated allocation quality by assuming that \textit{selecting} successful entries is sufficient. We now test \textsc{Mesa}-guided defense in two settings: (1) with real runtime detectors and (2) against attackers that adapt to the protected set.

\subsubsection{Runtime defense integration}
\begin{figure}[t]
\centering
\includegraphics[width=0.85\columnwidth]{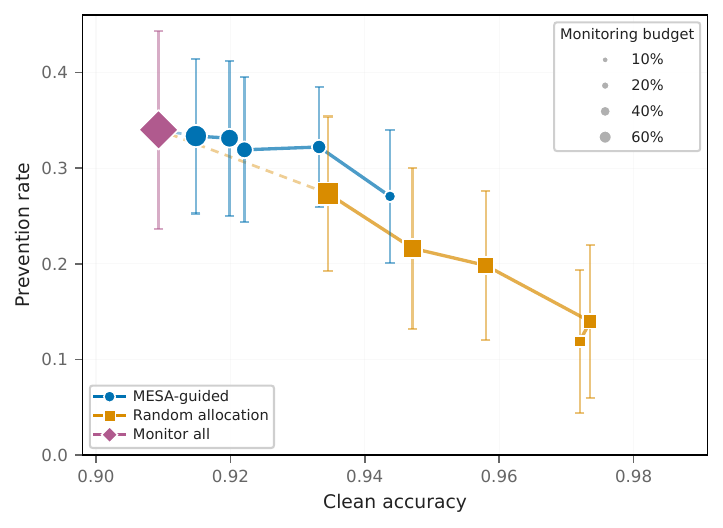}
\caption{Runtime NLI-defense security--utility frontier in CS. Each point represents one defense allocation.}
\label{fig:security-utility}
\end{figure}
For runtime defense integration, we place an evidence-grounded natural language inference (NLI) detector~\cite{tangMiniCheckEfficientFactChecking2024,zhaAlignScoreEvaluatingFactual2023} in the message-delivery path. On selected channels, it tests whether claims are supported by the receiver's evidence from other sources. If the content is flagged, it is replaced with a neutral placeholder, and the remainder of the workflow is regenerated rather than replayed. 

Importantly, runtime defense must balance three competing objectives: security, clean-task utility, and deployment cost. Monitoring more channels may prevent more tasks, but it also increases detector invocations and latency while exposing more benign messages to false-positive blocking. We therefore use the fraction of monitored edges as a deployment-cost proxy (the exact runtime cost also depends on the underlying model, message frequency, etc.) and compare allocation policies on a security--utility plane. The Pareto frontier contains policies for which neither prevention nor utility can improve without sacrificing the other.

\autoref{fig:security-utility} shows this tradeoff across allocation budgets. Error bars show SEM across (model, topology) pairs (details can be found in Appendix~\ref{app:defense-detail}). \textsc{Mesa} first improves sharply between 10\% and 20\% coverage. At 20\%, it achieves 32.2\% prevention, compared with 14.0\% for random placement. This is a gain of 18.2 percentage points on security at the cost of approximately 4.1 points of accuracy (maintaining 93.3\%). Afterwards, inspecting more messages reduces clean accuracy and increases operational cost without security gain. On the other hand, random allocation improves gradually as more channels are protected but remains below \textsc{Mesa} at every matched budget. Even though it retains higher accuracy, its substantially lower prevention indicates that it monitors less consequential communication. Furthermore, monitoring every edge increases prevention only to $34.0\%$, while lowering utility to 90.9\%. The 20\% \textsc{Mesa}-guided policy attains about 95\% of monitor-all prevention while inspecting only one-fifth of the channels and preserving 2.4 pp higher accuracy. Overall, rather than eliminating the security--utility--cost tradeoff, \textsc{Mesa} concentrates enforcement where each additional monitored channel provides the greatest security benefit.

\subsubsection{Against allocation-aware attacks}
\label{rq3-adaptive-attacks}
Next, we stress-test the placement policy itself. A deterministic top-$k$ policy can make the protected set predictable. We therefore evaluate against both a gray-box attacker, who knows the \textsc{Mesa} ranking, and a white-box, allocation-aware attacker, who is additionally aware of the protected channels. We set defender budget $k_d=2$ and attacker budget $k_a=3$.

\begin{figure}[t]
\centering
    \includegraphics[width=0.9\linewidth]{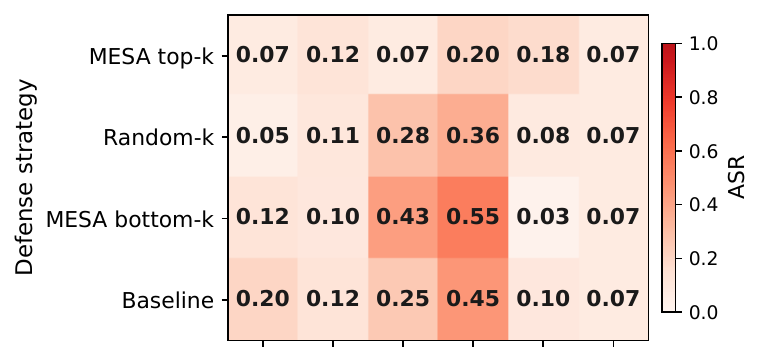}
    \label{fig:gray_box_attacker}
    \includegraphics[width=0.9\linewidth]{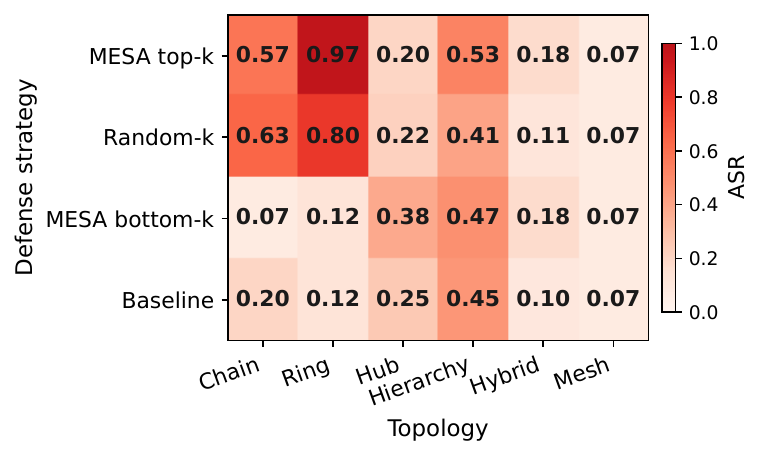}
    \label{fig:white_box_attacker}
\caption{ASR of gray-box (top) and white-box (bottom) attackers with $k_a=3$ budgets attacking a defender with a budget $k_d=2$.}
\label{fig:attacker_regimes}
\end{figure}

Against the gray-box attacker in \autoref{fig:attacker_regimes}, \textsc{Mesa} top-$k$ defense is highly effective in all topologies with low ASR, $0.07$--$0.20$. Here, the attacker and defender both target the same salient edges, and the \textsc{Mesa}-top defense remains effective even when the defender protects fewer edges than the attacker corrupts (\textit{i.e.}, $k_a > k_d$). Additionally, it reaches lower ASR compared to other defense strategies. The results illustrate the predictability and effectiveness of the \textsc{Mesa} metric even when more edges are attacked than defended.  

Furthermore, an allocation-aware attacker instead targets the highest-ranked unprotected edges, with ranks $k_d+1$ through $k_d+k_a$. This knowledge can nearly circumvent deterministic defense placement: top-$k$ defense reaches ASR 0.97 in \autoref{fig:attacker_regimes}, exposing the system as if \textsc{Mesa} had never been applied. Again, dense \texttt{Hybrid} and \texttt{Mesh} remain more resilient (approximately 0.14 and 0.07) because their redundant layouts dilute the effect of any small entry set.

A fixed top-$k$ policy is brittle because a fully informed attacker can evade it. We therefore test \textit{allocation diversification}: the defender keeps budget $k_d$ fixed but samples the protected set from the high-ranked edges before each run. The allocation-aware attacker knows the ranking rule but not the realized protected set. We compare fixed top-$k$, a uniform baseline over all edges, sampling from the top-three or top-four pool, and mixed policies that always protect the top edge while sampling the second. We use $k_d=2$, $k_a=3$, three brittle sparse topologies (\texttt{Chain}, \texttt{Ring}, and \texttt{Hub}), two models (Gemma-E4B and Qwen-9B), and three allocation seeds.

\begin{table}[t]
\centering
\small
\caption{Randomized defense against gray-box and white-box adaptive attackers.}
\label{tab:randomized-monitoring}
\setlength{\tabcolsep}{4pt}
\begin{tabular}{lcc}
\toprule
Defense policy & Gray-box ASR & Adaptive ASR \\
\midrule
\rowcolor{gray!10}
Fixed top-$k$        & \best{9.2}\,\%   & 58.3\,\% \\
Uniform baseline           & 14.7\,\% & 55.0\,\% \\
Sample top-3 pool               & 11.4\,\% & 58.1\,\% \\
\rowcolor{gray!10}
Sample top-4 pool               & 12.8\,\% & \best{47.8}\,\% \\
Pin top-1; sample top-3               & 12.8\,\% & 59.4\,\% \\
Pin top-1; sample top-4               & 10.0\,\% & 55.3\,\% \\
\bottomrule
\end{tabular}
\end{table}

\autoref{tab:randomized-monitoring} reports ASR averaged across the six topology-model pairs. Fixed top-$k$ is best against the gray-box attacker with $9.2\%$ ASR but degrades to 58.3\% against the allocation-aware attacker. Sampling from the top-four pool provides the best balance: gray-box ASR remains 12.8\%, while allocation-aware ASR falls to 47.8\%. The top-three pool is too predictable, whereas the uniform baseline sacrifices too much ranking information. Top-pool sampling is therefore distinct from the uniform baseline: it introduces uncertainty while retaining \textsc{Mesa}'s high-risk prior.

\resultbox{}{\textsc{Mesa}-guided NLI placement approaches monitor-all prevention with fewer inspections and higher clean utility; randomizing high-ranked edges before defense placement improves robustness to allocation-aware attackers.}

\subsection{RQ4: \textsc{Mesa} Practicality}
In the previous evaluation, we found that \textsc{Mesa} is effective and improves both offensive and defensive security analysis. We now assess how practical and readily deployable \textsc{Mesa} is. There are two main deployment constraints: (1) attacks and/or defenses can only take place on certain parts of the graph and (2) an allocation method should not depend on a particular orchestration implementation. In this section, we evaluate exposure-aware allocation and port workflows to LangGraph.

\subsubsection{Exposure-aware threat surface}
\label{sec:rq5-exposure}

The experiments above let the attacker target any edge. In deployment, edges differ in how much reachability they require. We assign each edge an exposure tier based on the source agent's role: \textit{direct}, \textit{tool-low}, \textit{tool-high}, and \textit{internal}. These tiers separate external-facing channels from edges that require stronger attacker capabilities, as summarized in \autoref{tab:exposure_tiers}.

\begin{table}[t]
\centering
\footnotesize
\caption{Edge exposure tiers.}
\label{tab:exposure_tiers}
\setlength{\tabcolsep}{4pt}
\renewcommand{\arraystretch}{1.12}
\begin{tabular}{@{}p{1.25cm}p{2.35cm}p{3.25cm}@{}}
\toprule
\textbf{Tier} & \textbf{Source edge} & \textbf{Required capability} \\
\midrule
direct
& \texttt{customer facing}
& External user input directly shapes the message. \\
tool-low
& \texttt{database}
& Write access to user-controlled records. \\
tool-high
& \texttt{policy}
& Write access to curated KB or RAG content. \\
internal
& \texttt{supervisor}, \texttt{manager}, \texttt{transaction}
& MITM access or compromise of an internal agent. \\
\bottomrule
\end{tabular}
\end{table}

\begin{figure}[t]
\centering
  \includegraphics[width=0.9\columnwidth]{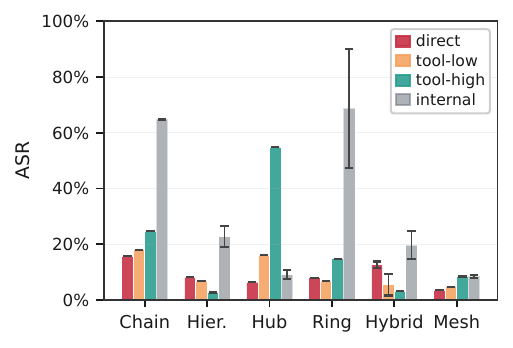}
\caption{ASR by exposure tier. Internal edges are most damaging; reachable edges have lower ASR.}
\label{fig:edge-exposure}
\end{figure}

\autoref{fig:edge-exposure} separates intrinsic damage from realistic reachability. Internal edges often have the highest intrinsic ASR ($\approx 70\%$ on \texttt{Ring}) but require the strongest attacker. More reachable edges such as \textit{direct} and \textit{tool-low} have non-negligible ASR at $5$--$18\%$. The \textit{tool-low} tier is more vulnerable than \textit{direct}: downstream agents treat database reports as facts, while customer-facing content is treated with more skepticism.

\begin{table}[t]
\centering
\small
\caption{Exposure-aware \textsc{Mesa} defense on realistic attack surfaces in the Customer Service scenario.}
\label{tab:exposure_defense}
\setlength{\tabcolsep}{4pt}
\renewcommand{\arraystretch}{1.08}
\begin{tabular}{@{}lcccccc@{}}
\toprule
 & \texttt{Hub} & \texttt{Hier.} & \texttt{Chain} & \texttt{Hybrid} & \texttt{Ring} & \texttt{Mesh} \\
\midrule
Baseline    & 16\% & 9\% & 5\% & 5\% & 4\% & 4\% \\
\textsc{Mesa} & 10\% & 9\% & 3\% & 5\% & 4\% & 4\% \\
\bottomrule
\end{tabular}
\end{table}

We then evaluate exposure-aware defense. MAS operators usually know which agents connect to users, retrieval stores, or internal services as auxiliary edge-exposure information. Here, the attacker samples one reachable edge per task, and the defender protects the top reachable edge by \textsc{Mesa}. As shown in \autoref{tab:exposure_defense}, the relatively strongest effect is in \texttt{Hub}, where ASR drops from $16.3\%$ to $10.4\%$. These gains are modest in absolute values because the reachable surface is smaller and less vulnerable than the full internal-edge surface, but they show that \textsc{Mesa} remains useful when it is restricted to edges an attacker can realistically reach.

\subsubsection{Portability to production frameworks}
\label{sec:generalization}
\textsc{Mesa} scores communication edges rather than features of our particular runtime framework. To test portability, we re-implement the customer-service workflow in \texttt{LangGraph}~\cite{LangGraphOverview} with the same roles, tasks, and topologies, evaluated on \texttt{Gemma-4-E4B} and \texttt{Qwen-3.5-9B}. \autoref{fig:langgraph} shows \textsc{Mesa} ranks from our runtime framework track \texttt{LangGraph} ASR. Across five non-\texttt{Mesh} topologies, Spearman $\rho$ between \textsc{Mesa} scores and \texttt{LangGraph} ASR is $+0.713$ for Gemma-4-E4B and $+0.863$ for Qwen-3.5-9B (both $p<10^{-4}$). \texttt{Mesh} is excluded because both runtimes produce ASR $\approx 0$. \textsc{Mesa} scores remain highly correlated with \texttt{LangGraph}'s empirical ASR, confirming the signal is not specific to our orchestration.

\begin{figure}[t]
\centering
 \includegraphics[width=0.95\columnwidth]{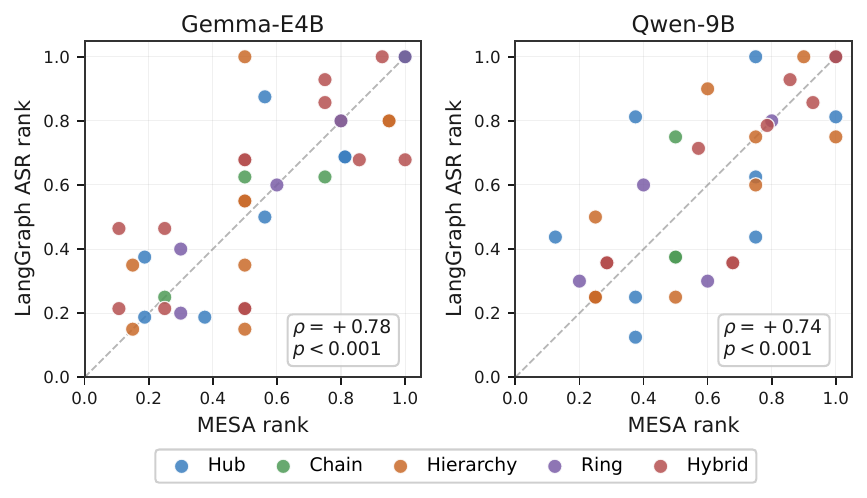}
\caption{\textsc{Mesa} edge scores transfer to \texttt{LangGraph}. Each point is one directed edge; per-topology $\rho \in [+0.53, +0.97]$, where $\rho$ is the Spearman coefficient
between the \textsc{Mesa} score and the empirical per-edge ASR.}
\label{fig:langgraph}
\end{figure}

\resultbox{}{\textsc{Mesa} remains useful under realistic reachability constraints and is readily portable to production frameworks, demonstrating practical utility.}
\section{Related Work}
\label{sec:bg-topology-defense}
\noindent\textbf{Topology-level MAS security.}
Previous work identifies communication topology as a system-level security factor. NetSafe~\cite{yuNetSafeExploringTopological2024} compares graph designs and demonstrates their differences in robustness. Subsequent work~\cite{shenUnderstandingInformationPropagation2025,liangTippingDominosTopologyAware2025} studies how benign and corrupted information propagates through those designs. TAMAS~\cite{kavathekarTAMASBenchmarkingAdversarial2025} and MASPi~\cite{anMASpiUnifiedEnvironment2025} further benchmark adversarial behavior across MAS roles and communication structures. Their topology-level findings are consistent with our preliminary analysis: workflows with bottlenecks and limited communication redundancy exhibit greater system-level susceptibility than workflows with redundant paths. In comparison, \textsc{Mesa} addresses a finer-grained problem: rather than selecting a topology or comparing whole systems, it ranks directed communication channels \emph{within} a fixed workflow to determine where limited security resources should be deployed. Broader security evaluations such as AgentDojo~\cite{debenedettiAgentDojoDynamicEnvironment2024} and A2ASecBench~\cite{liA2ASecBenchProtocolAwareSecurity2025} measure attacks through task or environment outcomes. \textsc{Mesa} follows the same outcome-based principle in evaluation but varies the attack intervention entry point to measure channel-specific security impact.

\noindent\textbf{Runtime monitoring and guardrails.}
Several runtime defenses identify unsafe behavior during execution and determine when to intervene. At the agent level, G-Safeguard~\cite{wangGSafeguardTopologyGuidedSecurity2025} trains a GNN anomaly detector on utterance graphs generated under attack and intervenes on agents classified as malicious. BlindGuard~\cite{miaoBlindGuardSafeguardingLLMbased2025} instead learns malicious-agent detection from normal executions and synthetic corruption, avoiding attack-specific labels. Beyond learned anomaly detection, SentinelAgent~\cite{heSentinelAgentGraphbasedAnomaly2025} applies policy- and context-based oversight to execution graphs, while SafeAgents~\cite{aroraExposingWeakLinks2025} uses adversarial trajectories to diagnose weak agents and architectural choices. Other systems add guard agents or constrain unsafe control flow~\cite{wuMonitoringLLMbasedMultiAgent2025,hossainMultiAgentLLMDefense2025,chennabasappaLlamaFirewallOpenSource2025,shiSAIGuardCommunicationStateSimulation2026}. Collectively, these methods predict or detect whether an agent, message, or action is malicious. \textsc{Mesa} predicts a complementary quantity: the expected downstream security impact if a communication channel is compromised, without using attack traces or labels. Such detectors can therefore be applied selectively to \textsc{Mesa}-ranked channels for security-utility benefits. Our NLI experiment demonstrates this end-to-end composition. 
\section{Discussion}
\label{sec:discussion}
\noindent\textbf{Declared workflows and ephemeral agents.}
\textsc{Mesa} targets operator-declared workflows with persistent communication identities, as in LangGraph, CrewAI, and other production frameworks~\cite{LangGraphOverview,CrewAIDocumentationCrewAI,wuAutoGenEnablingNextGen2023}. A declared graph does not mean every execution follows the same path: these graphs may contain branches, cycles, and conditional routing. \textsc{Mesa}'s workflow-aware features measure the impact of the channels activated during clean executions. The central requirement is that channel identities persist across executions long enough to estimate and act on their risk. Extending \textsc{Mesa} to ephemeral and ad hoc routes requires scoring reusable route templates or updating scores online as the workflow evolves.

\noindent\textbf{Redundancy dilutes channel criticality.}
Channel prioritization is most useful when information or authority is partitioned across roles (as in customer service). Its advantage diminishes when multiple agents or edges provide the same information. Mesh and the two random topology families introduce structural redundancy, while software-engineering agents share the task specifications and homogeneous debate agents share both the question and role. Here, a corrupted message competes with independently available context. Thus, security impact is distributed more evenly, leaving less opportunity for the estimated ranking to outperform uniform allocation. In contrast, bottlenecks that carry unique evidence produce the concentrated risk that \textsc{Mesa} is designed to identify. Redundancy is therefore not merely a topological property; rather, it also determines what agents know and whether the workflow can reconstruct missing or corrupted evidence.

\noindent\textbf{Deployment and beyond inter-agent channels.}
Deploying \textsc{Mesa} requires a declared workflow graph and representative clean tasks. Its ranking can guide runtime monitors, red-teaming, human review, and logging, and should be refreshed after changing individual system components. Specific defense budgets can be adapted to task sensitivity and system load. Furthermore, production workflows also contain security-relevant information flow beyond inter-agent communication, such as tool outputs and memory retrieval. Defenses at these boundaries can protect these flows before their evidence is transformed and enters inter-agent communication channels. Such defenses are complementary to \textsc{Mesa}'s defense allocation. Together, \textsc{Mesa} serves as a maintainable allocation layer rather than a temporary experimental score.

\section{Conclusion}
We presented \textsc{Mesa}, an allocation layer that ranks inter-agent communication channels by expected security impact before attack traces from target workflows are available. By combining graph-theoretic features with workflow-aware execution probes, \textsc{Mesa} identifies concentrated attack impact more effectively than structural or uniform allocation alone. With integrated NLI detectors, monitoring 20\% of channels achieves a 32.2\% prevention rate---approximately 95\% of monitor-all prevention---while preserving higher clean utility. The ranking also supports targeted security testing, remains useful under reachability constraints, and transfers to LangGraph. These results show that securing MAS requires studying not only individual agents or detectors, but also where consequential information moves. \textsc{Mesa} provides a practical basis for allocating limited security resources along those flows.

\section*{Acknowledgments}
The authors would like to thank Blaine Hoak for editorial feedback and manuscript improvement and Tingwei Zhang for helpful discussions and comments. 

\noindent\textbf{Funding acknowledgement.} This material is based upon work supported by the National Science Foundation under Grant No. CNS-2343611 and by the U.S. Army Research Office under MURI grant No. W911NF-21-1-0317. Any opinions, findings, and conclusions or recommendations expressed in this material are those of the author(s) and do not necessarily reflect the views of the National Science Foundation or the U.S. Army Research Office. The U.S. Government is authorized to reproduce and distribute reprints for government purposes notwithstanding any copyright notation hereon.

\appendix
\section*{Ethical Considerations}
This work studies integrity attacks in locally executed research workflows and does not target deployed services, real users, or production data. The controlled corruption intervention is used to compare entry-channel impact under a fixed adversarial objective. It should not be interpreted as an endorsement of attacking communication infrastructure. We report both prevention effectiveness and drops in the clean-utility of defense mechanisms. Released artifacts do not contain credentials but only synthetic and public benchmark tasks.

\section*{Open Science}
The artifact at \url{https://anonymous.4open.science/r/mesa/} contains source code, task and topology configurations, analysis scripts, and the aggregated tables needed to reproduce the paper’s results. The appendix distinguishes the broad execution sweep from the smaller feature-complete ranking and paired enforcement evaluation settings.

\bibliographystyle{IEEEtran}
\bibliography{references}
\section{Evaluation Design and Scope}
\label{app:evaluation-scope}
\noindent\textbf{Task construction.}
The customer-service suite contains 20 synthetic cases represented as structured customer, order, and policy records. We manually verified each case for internal consistency, unambiguous approve/deny decisions, and the required downstream actions. Cases cover standard returns, expired windows, membership exceptions, warranty-only remedies, non-refundable products, bulk orders, fraud flags, and international orders; none contains production or personal data. The software-engineering suite uses 50 public HumanEval tasks. A disjoint, outcome-blind difficulty-shift suite contains HumanEval problems whose solutions have 5--13 non-comment lines; it was frozen before attack evaluation and is reported separately. The debate suite uses public GSM8K and CommonsenseQA questions.

\noindent\textbf{Execution matrix.}
Each study executes the Cartesian product of its declared task, model, and topology sets and, for single-entry experiments, candidate edges. Combinations excluded in advance due to resource limits are marked not run and are neither imputed nor replaced by results from another study. Further evaluations use a subset of the overall sweep settings: breadth measures variation in attack impact; ranking requires complete fold-specific clean probes; enforcement uses paired executions with and without withholding; and portability reimplements the same workflow in another orchestration framework. \autoref{tab:study-scope} summarizes these matrices.

\begin{table*}[t]
\centering
\footnotesize
\caption{Evaluation matrices. A configuration is one
scenario--model--topology tuple; products are Cartesian within each row unless
an unavailable combination is marked not run.}
\label{tab:study-scope}
\setlength{\tabcolsep}{4.5pt}
\renewcommand{\arraystretch}{1.12}
\begin{tabularx}{\textwidth}{@{}>{\raggedright\arraybackslash}p{1.65cm}>{\raggedright\arraybackslash}p{3.50cm}>{\raggedright\arraybackslash}p{4.25cm}>{\raggedright\arraybackslash}X@{}}
\toprule
\textbf{Evaluation} & \textbf{Task and configurations} & \textbf{Repetition and controls} & \textbf{Question answered} \\
\midrule
Breadth & CS, SE, and HD; up to five models; six named and two random topologies & Deterministic clean and one-edge intervention runs & How impact changes across tasks, models, and topologies \\
Ranking & CS and SE; $3$ models $\times$ $5$ named topologies & 2 frozen task folds; each fold's features rank the other fold, then directions are collapsed & Whether attack-free features prioritize consequential edges \\
Enforcement & CS; $2$ models $\times$ $5$ named topologies & Paired control runs; random placement averaged over 5 fixed seeds & Whether guided defense prevents failures without monitor-all cost \\
Stress tests & CS; sparse-topology and multi-entry subsets & 5 seeds for multi-entry attacks; 3 seeds for diversified allocation & Whether ranking survives coordinated or allocation-aware attacks \\
Portability & CS; $2$ models $\times$ $5$ named topologies & Same roles, tasks, and graphs reimplemented in LangGraph & Whether rankings transfer across orchestration frameworks \\
\bottomrule
\end{tabularx}
\end{table*}

\noindent\textbf{Seeds and statistical independence.}
Agent inference and the customer-service judge use temperature zero. The intervention model uses temperature 0.3 to vary wording. Seeds are used only for randomized procedures (reported in \autoref{tab:study-scope}). Seeded results are averaged within each configuration. Ranking uses the configuration, not individual edges or tasks. Cross-fit directions are aggregated before resampling.

\noindent\textbf{Extensibility.}
The framework separates scenario definitions (roles, prompts, tasks, and outcome checker) from topology, model, intervention, and analysis modules. Therefore, a new application requires a runnable task suite with a reproducible success criterion and role-specific context, but no change is needed for the ranking pipeline. The current domains span distinct information-partition patterns. 

\section{Feature Selection}
\label{app:feature-audit}

\subsection{Feature activity}
A feature affects ranking only when it varies across edges in a configuration. Symmetric topologies create more ties, while hub, hierarchy, and hybrid graphs provide greater structural variation. Workflow-aware features vary in almost every configuration and therefore break some structural ties. Adding them to the six graph features improves attack coverage as shown in \autoref{fig:defense-coverage}. We retain all graph-theoretic features with no extra model call, even when they are inactive, to avoid post-hoc redefinition. Semantic non-recoverability (NR) is filtered out because its measurement confounds recoverability with the number of alternative contexts the receiver gets. 

\subsection{Feature redundancy}
\begin{figure*}[t]
\centering
\includegraphics[width=0.7\textwidth]{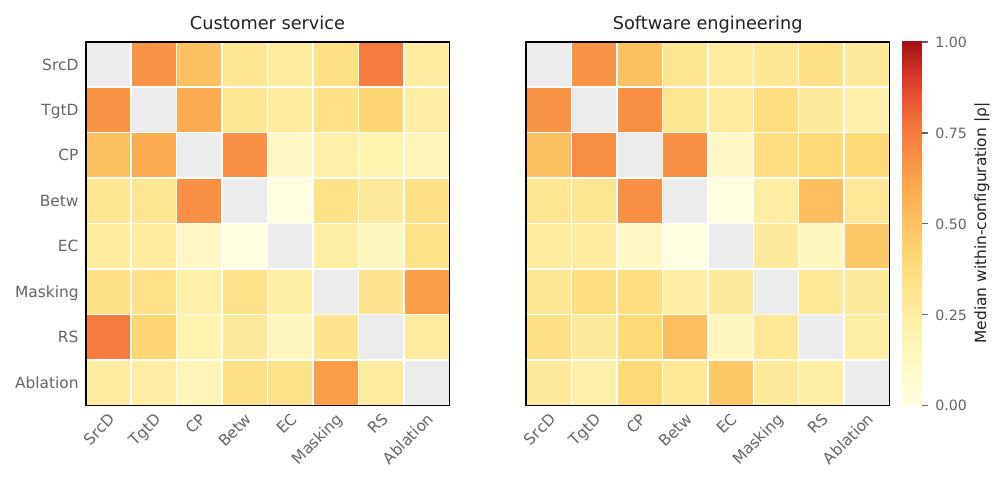}
\caption{Median within-configuration absolute Spearman correlation
($|\rho|$) among active candidate features. Values measure feature redundancy, not association
with attack success.}
\label{fig:feature-correlation}
\end{figure*}
\autoref{fig:feature-correlation} shows pairwise redundancy among active structural and workflow-aware features. Since information bottleneck and bridge indicator features are inactive for multiple topologies, they are excluded here. Each cell is the median absolute within-configuration Spearman correlation across configurations. Absolute correlation measures similarity in edge ranking regardless of direction. It does not measure association with intervention success. Here, the strongest associations occur among related structural quantities or probes that focus on similar workflow behaviors. Their locations differ by application domain. No off-diagonal pair is consistently near-duplicate in both domains, and no pair exceeds the $|\rho|=0.8$ threshold. We therefore retain all features here. The held-out comparison can be found in \S\ref{sec:rq1-mesa} to determine whether a probe can justify its measurement cost.  

\section{Software-Engineering Results}
\label{app:se-results}
\begin{figure}[t]
\centering
\includegraphics[width=0.85\columnwidth]{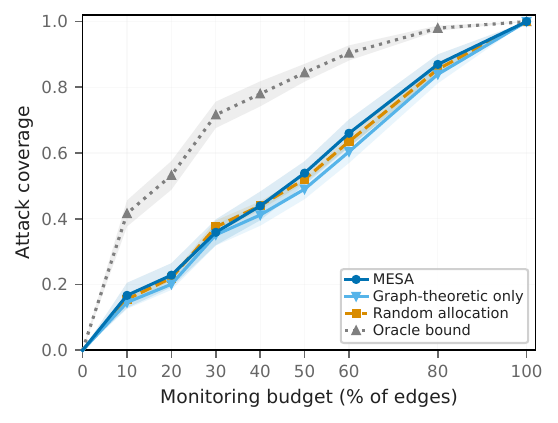}
\caption{Software-engineering intervention attack coverage. \textsc{Mesa} and graph-only ranking track random allocation, while the oracle shows remaining prioritization headroom.}
\label{fig:ranking-coverage-se}
\end{figure}
\autoref{fig:defense-coverage} reports the customer-service allocation curve. Here, the software-engineering results are shown in \autoref{fig:ranking-coverage-se}. At 10\% and 20\% budgets, \textsc{Mesa} exceeds random by only 1.2 percentage points. The graph-only curve shows the similar trend. Both remain well below the oracle curve. Recall that all agents receive the same task specification for SE tasks, where they can cross-check the original instructions. As discussed in \S\ref{sec:prelim-topologies}, this bounds the strongest allocation claim to role- and information-asymmetric workflows.

\section{Details on Runtime-Defense}
\label{app:defense-detail}
\begin{table*}[t]
\centering
\footnotesize
\caption{Runtime NLI-defense integration (CS). Adjusted prevention is paired against the matched no-withheld control; brackets show 95\% configuration-bootstrap CIs.}
\label{tab:defense-detail}
\setlength{\tabcolsep}{4.2pt}
\renewcommand{\arraystretch}{1.08}
\begin{tabular}{@{}lcccccc@{}}
\toprule
\textbf{Policy} & \textbf{Budget} & \textbf{Adjusted prevention} & \textbf{Clean utility} & \textbf{Message FPR} & \textbf{Detector calls} & \textbf{Latency (s/config.)} \\
\midrule
\textsc{Mesa} & 20\% & 0.322 [0.210, 0.442] & 0.933 & 0.141 & 3,105 & 1,637 \\
\textsc{Mesa} & 40\% & 0.331 [0.177, 0.477] & 0.920 & 0.142 & 6,008 & 2,316 \\
Random (5 seeds) & 20\% & 0.140 [0.001, 0.294] & 0.974 & 0.148 & 2,925 & 1,169 \\
Random (5 seeds) & 40\% & 0.216 [0.074, 0.384] & 0.947 & 0.148 & 5,944 & 1,985 \\
Monitor all & 100\% & 0.340 [0.156, 0.535] & 0.909 & 0.140 & 13,990 & 3,897 \\
\bottomrule
\end{tabular}
\end{table*}

\autoref{tab:defense-detail} shows the detailed data for \autoref{fig:security-utility}. The NLI contradiction threshold is set at $0.26$. All flagged messages are withheld and the remaining workflow is regenerated. Prevention is adjusted using a paired no-withholding rerun of the same trial. Confidence intervals resample the ten model--topology configurations, with random-allocation seeds averaged first. Detector calls and latency report actual execution cost. 

\section{Clean Baseline Accuracy}
\label{app:clean-accuracy}
\begin{table*}[t]
\centering
\footnotesize
\caption{Clean baseline accuracy by scenario, topology, and model.}
\label{tab:clean_accuracy}
\setlength{\tabcolsep}{4.5pt}
\renewcommand{\arraystretch}{1.08}
\begin{tabular}{@{}llccccc@{}}
\toprule
\textbf{Scenario} & \textbf{Topology} & \textbf{Llama-8B} & \textbf{Qwen-9B} & \textbf{Qwen-27B} & \textbf{Gemma-E4B} & \textbf{Gemma-26B} \\
\midrule
\textit{Customer Service}
& \texttt{Chain}     & 60\% & 100\% & 100\% & 90\%  & 95\%  \\
& \texttt{Ring}      & 60\% & 85\%  & 100\% & 95\%  & 100\% \\
& \texttt{Hub}       & 75\% & 95\%  & 95\%  & 100\% & 100\% \\
& \texttt{Hierarchy} & 40\% & 75\%  & 70\%  & 100\% & 85\%  \\
& \texttt{Hybrid}    & 50\% & 80\%  & 75\%  & 95\%  & 90\%  \\
& \texttt{Mesh}      & 50\% & 95\%  & 100\% & 95\%  & 75\%  \\
\midrule
\textit{Software Engineering}
& \texttt{Chain}     & 48\% & 82\%  & 86\%  & 74\%  & 84\%  \\
& \texttt{Ring}      & 60\% & 80\%  & 94\%  & 80\%  & 84\%  \\
& \texttt{Hub}       & 58\% & 82\%  & 92\%  & 80\%  & 88\%  \\
& \texttt{Hierarchy} & 64\% & 82\%  & 96\%  & 82\%  & 84\%  \\
& \texttt{Hybrid}    & 66\% & 78\%  & 94\%  & 78\%  & 86\%  \\
& \texttt{Mesh}      & 54\% & 82\%  & 94\%  & 80\%  & 84\%  \\
\midrule
\textit{Homogeneous Debate}
& \texttt{Chain}     & 87\% & 96\%  & --- & 92\% & 92\% \\
& \texttt{Ring}      & 89\% & 95\%  & --- & 93\% & 91\% \\
& \texttt{Hub}       & 89\% & 95\%  & --- & 93\% & 89\% \\
& \texttt{Hierarchy} & 88\% & 96\%  & --- & 92\% & 92\% \\
& \texttt{Hybrid}    & 88\% & 96\%  & --- & 92\% & 92\% \\
& \texttt{Mesh}      & 85\% & 96\%  & --- & 92\% & 92\% \\
\bottomrule
\end{tabular}
\end{table*}

\autoref{tab:clean_accuracy} reports clean accuracy for the six named topologies. Cells use all tasks for each domain, and dashes denote configurations not run due to the compute-time limit. Per-edge ASR is computed only over clean-correct tasks. 

\section{Random Topology Stress Test}
\label{app:random_topologies}
We use fixed-seed \textit{Erd\H{o}s--R\'enyi} ($n{=}7$, $p{=}0.45$) and \textit{Barab\'asi--Albert} ($n{=}7$, $m{=}2$) graphs as exploratory stress tests. They contain 22 and 20 directed channels, with undirected densities 0.52 and 0.48, respectively. Neither has a bridge, and both have edge connectivity two, providing greater route redundancy than the sparse role-structured workflows. With multiple alternative paths, single-edge impacts become more interchangeable and harder to rank.

The combined score raises mean $\rho$ from $+0.131$ to $0.190$, but most values are indistinguishable from zero. This serves as a boundary result. In these redundant graphs (similar to mesh), corrupting one channel is less likely to isolate unique information. \textsc{Mesa} is thus more useful when workflows provide meaningful role and collaboration heterogeneity.

\begin{table}[t]
\centering
\footnotesize
\caption{Edge-ranking agreement with ASR on the two redundant random graphs.
$^{**}p<0.01$; $^\dagger p\geq 0.05$.}
\label{tab:mesa_rho_random}
\setlength{\tabcolsep}{5pt}
\begin{tabular}{@{}lcc@{}}
\toprule
Model & Graph only & \textsc{Mesa} \\
\midrule
Llama-8B & +0.202$^\dagger$ & +0.102$^\dagger$ \\
Qwen-9B & $-$0.052$^\dagger$ & +0.235$^\dagger$ \\
Qwen-27B & +0.172$^\dagger$ & +0.123$^\dagger$ \\
Gemma-E4B & +0.276$^\dagger$ & +0.440$^{**}$ \\
Gemma-26B & +0.058$^\dagger$ & +0.048$^\dagger$ \\
\midrule
\textit{Mean} & +0.131{\scriptsize\,$\pm$0.052} & +0.190{\scriptsize\,$\pm$0.062} \\
\bottomrule
\end{tabular}
\end{table}

\end{document}